\documentclass[10pt]{article}

\usepackage[preprint]{neurips_2023}
\makeatletter
\renewcommand{\@noticestring}{}
\makeatother

\usepackage[utf8]{inputenc}
\usepackage[T1]{fontenc}
\usepackage{hyperref}
\usepackage{url}
\usepackage{booktabs}
\usepackage{amsfonts}
\usepackage{amsmath}
\usepackage{amssymb}
\usepackage{nicefrac}
\usepackage{microtype}
\usepackage{xcolor}
\usepackage{graphicx}
\usepackage{subcaption}
\usepackage{algorithm}
\usepackage{algorithmic}
\usepackage{enumitem}
\usepackage{amsthm}
\usepackage{ragged2e}
\usepackage{placeins}
\usepackage{fancyhdr}

\newtheorem{definition}{Definition}

\newtheorem{proposition}{Proposition}

\newtheorem{remark}{Remark}

\fancypagestyle{plain}{\pagestyle{fancy}}
\fancypagestyle{empty}{%
  \fancyhf{}
  
  \fancyfoot[C]{\footnotesize Preprint — version 1 (February 2026) \hfill 
\thepage}
}

\begin{document}


\title{
\textbf{Parametric Traversal for Multi-Dimensional Cost-Aware Graph Reasoning}\\[0.3em]
\large{A Policy-Driven Approach for Telecom and Datacenter Infrastructure Graphs}
}

\author{
Nicolas Tacheny \\
Ni2 Innovation Lab\\
\texttt{nicolas.tacheny@ni2.com}
}

\maketitle


\begin{abstract}
Classical path search assumes complete graphs and scalar optimization metrics,
yet real infrastructure networks are incomplete and require multi-dimensional
evaluation. We introduce the concept of \emph{traversal}: a generalization of
paths that combines existing edges with \emph{gap transitions}---missing but
acceptable connections representing links that can be built. This abstraction
captures how engineers actually reason about infrastructure: not just what
exists, but what can be realized.

We present a parametric framework that treats planned connections as first-class
transitions, scales to large graphs through efficient candidate filtering, and
uses multi-dimensional criteria to decide whether a traversal should continue
to be explored or be abandoned. We evaluate the framework through representative scenarios in datacenter
circuit design and optical route construction in telecommunication networks,
demonstrating conditional feasibility, non-scalarizable trade-offs, and
policy calibration capabilities beyond the reach of classical formulations.
\end{abstract}


\section{Introduction}

Network infrastructures typically consist of physical links between resources
(cables connecting devices and ports) and logical links between systems
(circuits connecting network nodes or hosts). A path is an ordered sequence of
nodes linked together to form an end-to-end connection between a source and a
destination. Finding optimal paths to connect devices and nodes is often
complex and traditionally relies on error-prone, manual, or cumbersome network
tools. Automated path finding has become a key operational capability for
digital service providers.

Path search is a foundational problem in graph theory and underlies numerous
applications in networking, logistics, robotics, and infrastructure management.
Classical formulations typically assume a fully specified graph and aim at
optimizing a simple metric such as path length or cumulative weight. However,
real-world infrastructure networks rarely conform to these assumptions.

In datacenters and telecommunication networks---whether wired or
wireless---infrastructures are often incomplete, heterogeneous, and subject to
operational constraints. Engineering decisions are fundamentally about what
\emph{can be built}, not just what exists. Certain connections may be missing
but inexpensive to create---a cable within the same rack, a planned fiber
extension, a microwave link between two towers---while others are costly or
forbidden except in exceptional circumstances. Design tools that ignore this
distinction and only traverse existing edges are of limited practical use,
as they cannot reason about how to meet new demand when existing connectivity
is insufficient.

Furthermore, the notion of an optimal path is highly contextual. Real engineers
evaluate feasibility by tracking multiple dimensions simultaneously: the number
of new connections to deploy, racks or sites crossed, optical attenuation
budgets, cost thresholds. Reducing all these factors to a single scalar weight
loses critical information and prevents meaningful optimization.

This paper introduces a programmable path search framework designed to address
these challenges. The framework is built around four key design decisions:

\paragraph{Gap transitions as first-class citizens.}
The framework models path search as a sequence of transitions between nodes,
explicitly distinguishing between \emph{edge transitions} (existing connections)
and \emph{gap transitions} (missing but acceptable connections). We call such
a sequence a \emph{traversal}---a generalization of the classical notion of
path that may combine existing edges with segments yet to be built. This
abstraction captures the reality of infrastructure planning where ``missing''
links are often cheap, planned, or policy-allowed. Traversal should be
understood as a planning abstraction rather than a routing artifact.

\paragraph{Separation of acceptability domain and predicate.}
To determine which gaps are acceptable, one could evaluate all $O(|N|^2)$ node
pairs---computationally prohibitive for large graphs. Instead, the framework
separates the problem into a fast, indexable \emph{acceptability domain} that
restricts candidates (using spatial indices, rack/site scoping, or vendor
compatibility rules) and a fine-grained \emph{acceptability predicate} for
detailed evaluation. This mirrors how production systems actually work and
enables scaling to large infrastructures.

\paragraph{Multi-dimensional accumulation state.}
Rather than scalarizing path costs into a single weight, the framework
associates each partial traversal with an arbitrary accumulation state---tracking
gap counts, racks crossed, attenuation budgets, or any combination thereof. This
preserves the information engineers need to make meaningful decisions.

\paragraph{Demonstrative evaluation on industrial scenarios.}
Rather than benchmarking against classical algorithms---which would require
reducing the model to a scalar formulation it is designed to
transcend---we evaluate the framework through representative scenarios drawn
from datacenter and telecommunication network design. These scenarios
demonstrate conditional feasibility under business rules, non-scalarizable
trade-offs, and ex-ante policy calibration: properties that classical
formulations cannot express.

\paragraph{Contributions.}
This paper makes the following contributions:

\paragraph{Traversal-level connectivity model.}
We introduce a formal model of connectivity based on traversals rather than
paths, enabling reasoning over both existing connections and potential
deployable transitions in incomplete infrastructure graphs.

\paragraph{Explicit separation of reasoning concerns.}
The framework cleanly separates candidate generation (acceptability domain),
feasibility filtering (acceptability predicate), resource accumulation, and
exploration policies, exposing independent control levers for operational
reasoning.

\paragraph{Policy-driven feasibility reasoning.}
We show how operational policies (e.g., deployment limits, physical budgets,
service tiers) can be encoded as exploration predicates, making feasibility an
explicit and inspectable outcome of the traversal process.

\paragraph{Support for non-scalarizable trade-offs.}
The traversal model naturally represents multiple admissible solutions that are
not totally ordered by a single cost metric, reflecting realistic engineering
trade-offs that cannot be captured by classical shortest-path formulations.

\paragraph{Demonstrative evaluation on telco and datacenter scenarios.}
Through structurally realistic scenarios, we demonstrate conditional
feasibility, non-comparable design alternatives, and policy calibration
capabilities that are not expressible with traditional graph-based routing
approaches.

\newpage

\section{Related Work}

Reasoning about connectivity has long been studied in graph theory and, more
recently, in networking research. However, these traditions have evolved
largely independently. On the one hand, theoretical work has progressively
enriched shortest-path search with multi-objective and resource constraints,
while separately studying how to augment graph connectivity. On the other hand,
networking research has advanced programmability and intent-based abstractions,
but often leaves the feasibility reasoning layer implicit. This section reviews
both lines of work and clarifies how the proposed traversal framework relates
to them.

\subsection{Graph-Theoretic Foundations: From Optimal Paths to Connectivity
Augmentation}

\paragraph{Shortest paths under constraints.}
Classical shortest-path algorithms assume a fully specified graph and optimize a
scalar objective. Over time, this formulation has been extended to incorporate
increasingly rich constraint structures.

Time-dependent shortest path problems already demonstrate that feasibility and
optimality may depend on external policies rather than solely on topology.
Dean~\cite{dean2004timedependent} studies minimum-cost paths in time-dependent
networks with waiting policies, showing how constraints on waiting behavior
alter the structure of feasible paths. In these formulations, however, the
graph remains fully defined; policies influence traversal cost or timing, not
the existence of transitions.

A major line of research generalizes shortest paths to multi-resource and
multi-objective settings. Resource-Constrained Shortest Path (RCSP) problems,
surveyed by Pugliese and Guerriero~\cite{pugliese2013rcsp}, associate multiple
resource accumulations with each path and rely on dominance-based label-setting
or label-correcting methods to prune infeasible or dominated solutions.
Multi-Objective Shortest Path (MOSP) problems further extend this perspective
by treating costs as vectors and computing Pareto-efficient path sets.
Foundational algorithmic developments include dominance-based extensions of
Martins' algorithm~\cite{gandibleux2006martins}, scalable parallel
label-setting approaches~\cite{sanders2013parallel}, and more recent
heuristic-search variants such as multi-objective D*
Lite~\cite{ren2022modlite}. A recent survey by Salzman et
al.~\cite{salzman2023mosp} synthesizes progress in heuristic multi-objective
search and highlights scalability challenges in high-dimensional cost spaces.

These contributions share two core assumptions: (1)~the graph topology is
explicitly given; (2)~search operates over existing edges only. The main
research questions concern optimality, dominance pruning, and scalability under
vector-valued costs. In contrast, the traversal framework proposed in this
paper retains multi-dimensional accumulation and pruning mechanisms but alters
the problem setting itself: transitions are not restricted to existing edges.
Instead, they may include acceptable but unrealized connections, whose
admissibility is governed by domain-specific policies. The focus shifts from
optimal path selection in a fixed topology to conditional feasibility in an
incomplete one.

\paragraph{Graph connectivity augmentation.}
A distinct branch of graph theory addresses incomplete connectivity through
augmentation. Graph augmentation problems study how to add a minimum set of
edges to satisfy connectivity properties such as $k$-edge or $k$-vertex
connectivity. Early approximation algorithms were developed by Frederickson and
Ja'Ja'~\cite{frederickson1981augmentation}, and comprehensive treatments appear
in surveys such as Frank and Jord\'{a}n~\cite{frank2015augmentation}. In this
setting, the objective is global: modify the graph so that it satisfies a
structural property with minimal cost.

While graph augmentation explicitly acknowledges missing edges, it treats
augmentation as a design decision applied to the entire graph. The output is a
modified topology satisfying a target property. In contrast, the traversal
framework does not compute a globally augmented graph. Instead, it explores
potential transitions locally and conditionally during search.
Missing-but-acceptable connections are considered as first-class transitions,
yet no permanent modification of the graph is performed. Feasibility emerges at
the level of individual traversals under operational constraints, rather than
as a global connectivity guarantee.

Thus, the proposed model can be understood as occupying an intermediate position
between classical path search (fixed topology) and graph augmentation (global
topology modification): it reasons over conditional transitions during search
without transforming the underlying graph structure.

\subsection{Networking Research: From Programmability to Intent and Feasibility}

\paragraph{Programmable and policy-driven networks.}
In parallel with theoretical developments, networking research has
progressively abstracted control and management. Software-defined networking
(SDN) and related paradigms decouple control logic from forwarding behavior,
enabling centralized decision-making and policy enforcement. Building on this
foundation, intent-based networking (IBN) aims to express high-level
operational objectives independently of their implementation.

RFC~9315~\cite{rfc9315} provides a formal clarification of intent concepts,
distinguishing intent from policy and outlining the functional components
required for intent realization. Recent surveys, such as Leivadeas and
Falkner~\cite{leivadeas2023ibn}, review architectural patterns for intent
translation, policy enforcement, and assurance in automated networks. These
works significantly advance the abstraction of network management and formalize
how desired states can be expressed and validated.

However, while intent-based frameworks define what the network should achieve,
they often assume that the underlying connectivity can be computed over a known
topology. The reasoning process that determines whether a connectivity
objective is physically realizable---especially in partially documented or
incomplete infrastructures---remains embedded within system-specific
implementations. Feasibility is typically treated as an outcome of routing or
provisioning tools rather than as an explicit modeling construct.

\paragraph{Routing with multiple constraints in operational networks.}
Routing with multiple constraints has been extensively studied in the context of
QoS provisioning and software-defined networking (SDN). Comprehensive surveys,
such as Guck et al.~\cite{guck2018qos}, review delay-constrained and
multi-constrained routing algorithms designed for centralized SDN control.
These approaches compute feasible paths under delay or bandwidth constraints
and improve expressiveness beyond scalar shortest-path routing.

However, they operate over explicitly defined network graphs in which all
admissible links are assumed to be known and represented in the topology. In
practical infrastructure planning scenarios---such as optical route
construction or datacenter cabling---the topology is often incomplete,
evolving, or only partially documented. Engineers frequently reason not only
about existing links but also about connections that could be deployed under
operational constraints. Traditional constraint-based routing formulations do
not model this distinction explicitly, as feasibility is evaluated over a fixed
topology. In contrast, our framework enables reasoning over conditionally
admissible transitions, making connectivity itself policy-dependent rather than
purely topological.

\subsection{Positioning of the Traversal Framework}

The traversal framework proposed in this paper integrates and extends these
lines of work along three dimensions.

First, like RCSP and MOSP approaches, it maintains multi-dimensional
accumulation states and relies on pruning predicates to control combinatorial
growth. However, it does not seek a single optimal solution or even a Pareto
frontier over a fixed topology; instead, it enumerates admissible traversals
whose feasibility depends on policy-defined constraints.

Second, unlike graph augmentation, it does not compute a globally modified
graph satisfying a structural property. Acceptable but unrealized connections
are treated as conditional transitions explored during search, without altering
the base graph. Connectivity is therefore a property of traversals under
policy, not a static property of the graph.

Third, while intent-based networking and policy-driven architectures define
high-level objectives and governance mechanisms, the traversal model provides
an explicit reasoning layer that bridges intent and physical realizability in
incomplete infrastructures. The separation between acceptability domain and
predicate, together with externalized accumulation and exploration predicates,
makes feasibility an inspectable and programmable construct rather than an
implicit by-product of routing.

In summary, prior work has enriched path search with multi-dimensional costs,
formalized connectivity augmentation, and abstracted network intent. The
traversal framework unifies these perspectives by introducing policy-governed
transition generation in incomplete typed graphs, enabling conditional
connectivity reasoning that is neither reducible to classical shortest-path
optimization nor to global graph augmentation.


\section{Parametric Traversal Model}

\subsection{Typed Graph}

Let $\mathcal{P}_N$ and $\mathcal{P}_E$ be arbitrary sets representing node
properties and edge properties, respectively.

\begin{definition}[Typed Graph]
A \emph{typed graph} is a tuple $G = (N, E, \rho_N, \rho_E)$ where:
\begin{itemize}[leftmargin=1.5em,nosep]
\item $N$ is a finite set of nodes,
\item $E \subseteq N \times N$ is a set of directed edges,
\item $\rho_N : N \rightarrow \mathcal{P}_N$ assigns properties to nodes,
\item $\rho_E : E \rightarrow \mathcal{P}_E$ assigns properties to edges.
\end{itemize}
\end{definition}

The property sets $\mathcal{P}_N$ and $\mathcal{P}_E$ are application-specific
and may encode physical, geographical, or operational characteristics.

\subsection{Transitions and Traversals}

A central feature of this framework is the ability to traverse not only 
existing edges but also acceptable connections between nodes that are not yet 
realized in the graph. We call such acceptable connections \emph{gaps}.

To determine which connections are acceptable, one could in principle evaluate
all $O(|N|^2)$ pairs of nodes. However, this is computationally prohibitive for
large graphs. Moreover, in practice, only a small subset of pairs is relevant.

To address this, we separate the problem into two stages:
\begin{enumerate}[leftmargin=1.5em]
\item The \emph{acceptability domain} restricts, for each node $n$, the set of
candidate nodes to consider. This restriction can leverage efficient data
structures such as spatial indices, type-based filtering, or precomputed
neighborhood structures.
\item The \emph{acceptability predicate} performs fine-grained evaluation on 
the restricted set of candidates, encoding complex business rules or 
constraints.
\end{enumerate}

\begin{definition}[Acceptability Domain]
For each node $n \in N$, the \emph{acceptability domain} $\gamma_n \subseteq N$
is the set of nodes to evaluate for acceptability from $n$.
\end{definition}

\begin{definition}[Acceptability Predicate]
An \emph{acceptability predicate} is a function
\[
\alpha : \{(n, m) \in N \times N \mid m \in \gamma_n\} \rightarrow
\{\mathrm{true}, \mathrm{false}\}
\]
that determines whether a connection between two nodes is acceptable.
\end{definition}

\begin{definition}[Transition]
Given a typed graph $G = (N, E, \rho_N, \rho_E)$, acceptability domains
$(\gamma_n)_{n \in N}$, and an acceptability predicate $\alpha$, a
\emph{transition} is a pair $t = (n, m) \in N \times N$ that is either:
\begin{itemize}[leftmargin=1.5em,nosep]
\item an \emph{edge transition} if $(n, m) \in E$, or
\item a \emph{gap transition} if $(n, m) \notin E$, $m \in \gamma_n$, and
$\alpha(n, m) = \mathrm{true}$.
\end{itemize}
We denote by $T$ the set of all transitions.
\end{definition}

\begin{definition}[Traversal]
A \emph{traversal} in a typed graph $G$ is a finite sequence of transitions
$\tau_k = (t_1, \ldots, t_k)$ where $t_i = (n_i, m_i)$ and $m_i = n_{i+1}$ for
all $i \in \{1, \ldots, k-1\}$. The integer $k \geq 0$ is the \emph{length} of
the traversal. We denote by $\mathcal{T}$ the set of all traversals.
\end{definition}

A traversal is said to be \emph{edge-connected} if all its transitions are edge 
transitions; in this case it corresponds to a classical path in the graph. A 
traversal containing at least one gap transition is called a \emph{gapped
traversal}.

\subsection{Accumulation}

In classical graph algorithms, paths are often evaluated using a single scalar
metric such as length or total weight. However, real-world applications
frequently require tracking multiple properties simultaneously: costs, risks,
resource consumption, constraint violations, or any combination thereof. Here,
\emph{costs} and \emph{risks} are understood in a broad sense---they may
represent monetary expenses, latency, energy consumption, failure probability,
or any domain-specific metric relevant to the application.

The \emph{accumulation} mechanism addresses this need by associating an 
arbitrary state with each traversal. This state is computed incrementally as 
the traversal extends, allowing the search algorithm to:
\begin{itemize}[leftmargin=1.5em,nosep]
\item evaluate and compare traversals based on complex, application-specific 
criteria,
\item prune branches early when the accumulated state indicates infeasibility,
\item return not only the traversal itself but also its associated properties.
\end{itemize}

Let $\mathcal{A}$ be an application-specific set of accumulation states (e.g.,
$\mathbb{R}^+ \times \mathbb{R}^+$ for costs and distances).

\begin{definition}[Accumulation]
An \emph{accumulation} over $\mathcal{A}$ is defined by:
\begin{itemize}[leftmargin=1.5em,nosep]
\item an initial state $a_0 \in \mathcal{A}$, and
\item a step function $g : \mathcal{A} \times \mathcal{T} \rightarrow \mathcal{A}$.
\end{itemize}
The accumulation of a traversal $\tau_k$ of length $k$ is denoted $a_{\tau_k}$
and defined recursively as:
\[
a_{\tau_k} =
\begin{cases}
a_0 & \text{if } k = 0 \\
g(a_{\tau_{k-1}}, \tau_k) & \text{if } k \geq 1
\end{cases}
\]
\end{definition}

The step function receives the previous accumulation $a_{\tau_{k-1}}$ and the
traversal up to step $k$. This allows implementations to use only the last
transition $t_k$ for simple cases, or to access the full history when needed.

\subsection{Exploration Predicate}

As the traversal extends, the algorithm must decide how to proceed: should it
continue exploring, terminate with a valid solution, or prune the current
branch? This decision depends on the traversal and its accumulated state, and
is application-specific.

\begin{samepage}
\begin{definition}[Exploration Predicate]
An \emph{exploration predicate} is a function
\[
\sigma : \mathcal{T} \times \mathcal{A} \rightarrow \{\mathtt{continue},
\mathtt{terminate}, \mathtt{prune}\}
\]
that determines the action to take for a traversal $\tau$ with accumulation
$a_\tau$:
\begin{itemize}[leftmargin=1.5em,nosep]
\item $\mathtt{continue}$: the traversal is extended by exploring successors,
\item $\mathtt{terminate}$: the traversal is added to the solution set and
exploration continues from its successors,
\item $\mathtt{prune}$: the traversal is discarded and its successors are not
explored.
\end{itemize}
\end{definition}
\end{samepage}


\section{Traversal Search Algorithm}

This section presents a generic algorithm for finding traversals in a typed graph, combining the components defined in Section~3: the typed graph $G = (N, E, \rho_N, \rho_E)$, the acceptability domain $\gamma_n$, the acceptability predicate $\alpha$, the accumulation function $g$, and the exploration predicate $\sigma$.

The algorithm explores the space of partial traversals. A \emph{traversal state} is a pair $(\tau, a_\tau)$ where $\tau \in \mathcal{T}$ is a traversal and $a_\tau \in \mathcal{A}$ is its accumulation. For a traversal $\tau_k = (t_1, \ldots, t_k)$ with $t_k = (n_k, m_k)$, we call $m_k$ the \emph{current node}.

The algorithm maintains a \emph{frontier} $\mathcal{F}$: a collection of traversal states awaiting expansion. The \emph{frontier policy} $\phi$ determines the order in which states are extracted (e.g., FIFO for BFS, priority queue for UCS or A*, bounded-width for beam search).

When a state $(\tau_k, a_{\tau_k})$ with current node $n$ is extracted, successor states are generated:
\begin{itemize}[leftmargin=1.5em,nosep]
\item \textbf{Edge transitions:} for each $m$ such that $(n, m) \in E$, form $\tau_{k+1} = \tau_k \cdot (n, m)$ with accumulation $a_{\tau_{k+1}} = g(a_{\tau_k}, \tau_{k+1})$.
\item \textbf{Gap transitions:} for each $m \in \gamma_n$ such that $(n, m) \notin E$ and $\alpha(n, m) = \mathrm{true}$, form $\tau_{k+1} = \tau_k \cdot (n, m)$ with accumulation $a_{\tau_{k+1}} = g(a_{\tau_k}, \tau_{k+1})$.
\end{itemize}
Here $\tau \cdot t$ denotes the concatenation of transition $t$ to traversal $\tau$.

\newpage
\subsection{Algorithm}
The complete procedure is given in Algorithm~\ref{alg:traversal}.

\begin{algorithm}[t]
\begin{algorithmic}[1]
\REQUIRE Typed graph $G = (N, E, \rho_N, \rho_E)$, start node $n_0$
\REQUIRE Acceptability domains $(\gamma_n)_{n \in N}$, acceptability predicate $\alpha$
\REQUIRE Accumulation function $g$, initial accumulation $a_0$
\REQUIRE Exploration predicate $\sigma$, frontier policy $\phi$
\ENSURE Set of valid traversals $\mathcal{S}$
\STATE $\tau_0 \leftarrow ()$ \COMMENT{Empty traversal}
\STATE $\mathcal{F} \leftarrow \{(\tau_0, a_0)\}$ \COMMENT{Initialize frontier}
\STATE $\mathcal{S} \leftarrow \emptyset$ \COMMENT{Solution set}
\WHILE{$\mathcal{F} \neq \emptyset$}
\STATE $(\tau, a_\tau) \leftarrow \textsc{Extract}(\mathcal{F}, \phi)$
\STATE $s \leftarrow \sigma(\tau, a_\tau)$
\IF{$s = \mathtt{prune}$}
\STATE \textbf{continue}
\ENDIF
\IF{$s = \mathtt{terminate}$}
\STATE $\mathcal{S} \leftarrow \mathcal{S} \cup \{(\tau, a_\tau)\}$
\ENDIF
\STATE $n \leftarrow \textsc{CurrentNode}(\tau)$
\STATE $V_\tau \leftarrow \textsc{VisitedNodes}(\tau)$ \COMMENT{Nodes in $\tau$}
\FOR{each $m$ such that $(n, m) \in E$ or ($m \in \gamma_n$ and $\alpha(n, m)$)}
\IF{$m \in V_\tau$}
\STATE \textbf{continue} \COMMENT{Avoid cycles}
\ENDIF
\STATE $\tau' \leftarrow \tau \cdot (n, m)$
\STATE $a_{\tau'} \leftarrow g(a_\tau, \tau')$
\STATE $\textsc{Insert}(\mathcal{F}, (\tau', a_{\tau'}), \phi)$
\ENDFOR
\ENDWHILE
\RETURN $\mathcal{S}$
\end{algorithmic}
\caption{Parametric Traversal Search. The algorithm explores traversal states
from a frontier, extending partial traversals through edge and gap transitions.
The exploration predicate $\sigma$ controls pruning and termination; the
accumulation function $g$ tracks multi-dimensional state along each traversal.}
\label{alg:traversal}
\end{algorithm}

The function $\textsc{CurrentNode}(\tau)$ returns the current node: $n_0$
if $\tau$ is empty, otherwise the target node of the last transition in $\tau$.
The function $\textsc{VisitedNodes}(\tau)$ returns the set of all nodes
appearing in the traversal $\tau$.

\subsection{Complexity}

Let $|N|$ denote the number of nodes and $|E|$ the number of edges. At a node
$n$, the algorithm considers at most $\deg(n) + |\gamma_n|$ successor
transitions, where $\deg(n)$ denotes the out-degree of $n$ in $G$.

In the worst case, the branching factor is bounded by
$\max_{n \in N} \deg(n) + \max_{n \in N} |\gamma_n|$. Assuming that the
algorithm explores all traversals up to a maximum length $L$, and using the
average out-degree
\[
\bar{d} = \frac{1}{|N|} \sum_{n \in N} \deg(n) = \frac{|E|}{|N|}
\]
as an approximation, the total number of traversal states explored is bounded by
\[
O\bigg(\Big(\frac{|E|}{|N|} + \max_{n \in N} |\gamma_n|\Big)^L\bigg),
\]
which highlights the combined impact of graph density and acceptability domain
size on the search complexity. Note that the cycle avoidance mechanism (line~11
of Algorithm~\ref{alg:traversal}) ensures that each traversal visits at most $|N|$ distinct nodes,
implicitly bounding $L \leq |N|$.

This exponential complexity makes the exploration predicate $\sigma$ and the accumulation function $g$ critical for practical performance. These functions encode domain-specific knowledge about what constitutes a viable traversal: the accumulation tracks relevant metrics (cost, resource consumption, constraint violations), while $\sigma$ determines when to prune, terminate, or continue.

Crucially, $\sigma$ and $g$ are external to the algorithm---they encapsulate business rules and domain expertise. A well-designed accumulation that captures the essential properties of traversal quality, combined with an exploration predicate that aggressively prunes unpromising branches, can reduce the effective complexity by orders of magnitude. Conversely, a poorly specified $\sigma$ that fails to prune early will result in exhaustive exploration and intractable runtime.

The framework thus places the responsibility for efficiency on the domain expert: the algorithm provides the exploration mechanism, but the user must define what makes a traversal good or bad through $g$ and $\sigma$. The framework does not aim at improving worst-case complexity bounds over classical graph search. Its contribution lies in making explicit where combinatorial growth originates and how domain knowledge can be used to control it.

\subsection{Termination}

\begin{proposition}[Termination under bounded accumulation]
Assume that:
\begin{enumerate}[leftmargin=1.5em]
\item the graph $G = (N, E)$ is finite,
\item for each node $n \in N$, the acceptability domain $\gamma_n$ is finite,
\item there exists a bound $L \in \mathbb{N}$ such that
$\sigma(\tau, a_\tau) = \mathtt{prune}$ for all traversals $\tau$ with length
greater than $L$.
\end{enumerate}
Then Algorithm~\ref{alg:traversal} terminates after a finite number of iterations.
\end{proposition}

\begin{proof}
Since $N$ and all acceptability domains $\gamma_n$ are finite, the number of
possible successors of any traversal state is finite. Under assumption (3), no
traversal of length greater than $L$ is expanded.

Therefore, the algorithm explores only traversals of length at most $L$. The
number of such traversals is finite, bounded by
\[
\sum_{k=0}^{L} \prod_{i=1}^{k} (\deg(n_i) + |\gamma_{n_i}|),
\]
where $\deg(n_i)$ denotes the out-degree of node $n_i$.

As a result, the frontier becomes empty after finitely many expansions, and the
algorithm terminates.
\end{proof}

\begin{remark}
In practice, the bound on traversal length is typically enforced indirectly
through the accumulation function and the exploration predicate, for instance by
limiting cost, resource consumption, or the number of gap transitions. This
allows termination to be guaranteed without imposing an explicit depth limit.
\end{remark}


\section{Demonstrative Evaluation}

This work is not evaluated through classical performance metrics, as it does
not propose a new path-finding algorithm nor claim improved optimality.
Instead, we focus on the expressiveness and structural properties of the
traversal model, supported by formal analysis and representative case studies.

Our evaluation presents three scenarios that expose the expressive limits of
classical formulations and demonstrate how traversal-level reasoning enables
conditional feasibility, non-scalarizable trade-offs, and policy calibration.
Each scenario carries a distinct message that cannot be reduced to a
benchmark comparison.

\subsection{Industrial Context}
 
We ground the evaluation in two infrastructure domains.

\emph{Datacenter.} Rooms contain rows, rows contain racks, and racks contain
devices. Client racks house servers and patch panels; distribution racks house
switches with upstream connectivity. Structured cabling interconnects patch
panels along each row toward distribution racks. Within a rack, cross-connects
are inexpensive to deploy; inter-rack cabling is costly and typically avoided.
The routing problem is to connect a customer server to upstream network
equipment, potentially through gaps representing cables to be deployed or
device reconfigurations to be performed (e.g., configuring a VLAN to bridge
two ports on a switch).

\emph{Telecommunication network.} Sites contain optical distribution frames
(ODFs), splice boxes, and amplifiers, interconnected by fiber segments. Each
fiber segment carries physical properties (length, attenuation). Intra-site
extensions are relatively inexpensive; inter-site gaps require civil works.
The routing problem is to establish an optical path between two ODFs while
respecting physical constraints such as attenuation budgets and distance
limits.

\emph{Broader applicability.} Although the scenarios presented here focus on
wired infrastructures, the traversal model applies equally to wireless domains.
In microwave radio-link networks, gap transitions can represent planned radio
links subject to line-of-sight and frequency licensing constraints. In 5G RAN
deployments, they can model potential fronthaul or midhaul connections between
distributed units and central units. In satellite networks, gaps may represent
inter-satellite or ground-to-satellite links subject to orbital geometry and
link budget constraints. In each case, the acceptability domain restricts
candidate links by physical feasibility, and the accumulation tracks
domain-specific metrics such as signal-to-noise ratio, latency, or spectral
efficiency.

\subsection{Conditional Feasibility in Incomplete Networks}

The first evaluation demonstrates that feasibility is not a binary property of
the graph, but a conditional property of a traversal under business rules.

\subsubsection{Telco --- Optical Route under Physical Constraints}

Consider connecting a source ODF to a target ODF in a partially documented
network. The following operational rules apply:

\begin{itemize}[leftmargin=1.5em,nosep]
\item Intra-site extensions (gaps) are authorized if the distance between
endpoints is less than 100\,m.
\item Total attenuation must not exceed 30\,dB.
\item At most 1 gap transition is allowed.
\item Amplifiers may be traversed at most once.
\end{itemize}

These rules are not independently scalarizable: attenuation depends on fiber
type and length, the gap limit is a hard count constraint, and the amplifier
rule introduces a history-dependent condition. No single weighted shortest-path
formulation can capture all four simultaneously without losing constraint
semantics.

\paragraph{Traversal parametrization.}
\begin{itemize}[leftmargin=1.5em,nosep]
\item $\gamma_n = \{m \in N \mid \mathrm{site}(m) = \mathrm{site}(n)\}$:
candidates restricted to the same site.
\item $\alpha(n, m) = \mathrm{true}$ iff
$\mathrm{dist}(n, m) < 100\,\text{m}$ and fiber types are compatible.
\item $\mathcal{A} = \mathbb{R}^+ \times \mathbb{R}^+ \times \mathbb{N}$
with $a_\tau = (\text{total length}, \text{total attenuation},
\text{number of gaps})$.
\item $\sigma(\tau, a_\tau) = \mathtt{prune}$ if attenuation $> 30\,\text{dB}$
or gaps $> 1$; $\mathtt{terminate}$ if the current node is the target ODF;
$\mathtt{continue}$ otherwise.
\end{itemize}

\paragraph{Result.}
A valid traversal is found that includes one intra-site gap, respects the
attenuation budget, and satisfies the amplifier constraint. The traversal is
returned with its full accumulation state, making the engineering rationale
explicit: the gap location, the attenuation contribution of each segment, and
the remaining budget are all visible in the output.

This routing problem cannot be expressed as a single weighted shortest path
without losing constraint semantics. The traversal framework solves it without
post-processing and without constraint relaxation.

\subsubsection{Datacenter --- Physical Circuit under Client Policy}

Consider connecting a customer server to an upstream switch. The following
rules apply:

\begin{itemize}[leftmargin=1.5em,nosep]
\item Cross-connects (gaps) are authorized within the same rack.
\item Inter-rack gaps are forbidden for standard clients.
\item At most 2 new cables may be deployed.
\item At most 1 row change is allowed.
\end{itemize}

\paragraph{Traversal parametrization.}
\begin{itemize}[leftmargin=1.5em,nosep]
\item $\gamma_n = \{m \in N \mid \mathrm{rack}(m) = \mathrm{rack}(n)
\text{ and } \mathrm{type}(m) \in \{\text{patch panel}\}\}$: candidates
restricted to the same rack.
\item $\alpha(n, m) = \mathrm{true}$ iff both $n$ and $m$ have available
ports.
\item $\mathcal{A} = \mathbb{N} \times \mathbb{N} \times \mathbb{N}$
with $a_\tau = (\text{gaps}, \text{racks traversed}, \text{rows crossed})$.
\item $\sigma(\tau, a_\tau) = \mathtt{prune}$ if gaps $> 2$ or rows $> 1$;
$\mathtt{terminate}$ if $\mathrm{upstream}(m_k) = \mathrm{true}$;
$\mathtt{continue}$ otherwise.
\end{itemize}

\paragraph{Result.}
Under a premium client policy (inter-rack gaps allowed, up to 5 gaps), a valid
traversal is found. Under a standard client policy (rack-local gaps only, at
most 2 gaps), the same query yields no admissible traversal.

The infrastructure is identical in both cases. The feasibility depends entirely
on the policy encoded in $\gamma$, $\alpha$, and $\sigma$---not on the graph
structure. This distinction is invisible to any formulation that reduces
feasibility to graph connectivity.

\subsection{Non-Scalarizable Trade-offs in Physical Circuit Design}

The second evaluation demonstrates that certain engineering decisions cannot be
ordered by a single scalar cost function.

\paragraph{Scenario.}
Consider a datacenter where a server-to-upstream connection admits two valid
traversals:

\begin{center}
\begin{tabular}{lccc}
\toprule
& Gaps & Racks traversed & Rows crossed \\
\midrule
Traversal $\tau_1$ & 0 & 2 & 1 \\
Traversal $\tau_2$ & 1 (intra-rack) & 1 & 0 \\
\bottomrule
\end{tabular}
\end{center}

Both traversals satisfy all operational constraints. Neither dominates the
other:

\begin{itemize}[leftmargin=1.5em,nosep]
\item $\tau_1$ requires no new cable but crosses a row boundary;
\item $\tau_2$ stays within a single rack and row but requires deploying one
cross-connect.
\end{itemize}

No fixed weighting of these three dimensions can consistently reflect all
possible operational policies. An operator prioritizing deployment cost would
prefer $\tau_1$; an operator prioritizing spatial locality would prefer
$\tau_2$. Collapsing the accumulation into a scalar erases this distinction.

\paragraph{What the framework provides.}
The traversal search returns both $\tau_1$ and $\tau_2$ as elements of the
solution set $\mathcal{S}$, each annotated with its full accumulation state.
The final selection can be deferred to a human operator or resolved by a
downstream policy function operating on the structured accumulation.

The output of traversal search is a decision space, not a single optimum.
This is precisely the regime where the multi-dimensional accumulation model
provides value over scalar-cost formulations.

\subsection{Policy Calibration and Severity Analysis}

The third evaluation demonstrates that the framework can serve as an
instrument for calibrating operational policies before deployment.

\paragraph{Scenario.}
A telecommunication operator wants to determine the attenuation budget $B$ at
which the network becomes predominantly connectable. The question is:

\emph{``Starting from what attenuation budget does the majority of ODF pairs
become reachable?''}

This is not a path-finding query but a structural analysis of the policy
itself.

\paragraph{Procedure.}
The exploration predicate is parameterized by $B$:
\[
\sigma_B(\tau, a_\tau) = \mathtt{prune} \quad \text{if total attenuation} > B
\]
For each value of $B \in \{20, 25, 30, 35\}\,\text{dB}$, we analyze which
source-target ODF pairs admit at least one valid traversal. All other
parameters ($\gamma$, $\alpha$, $g$) remain fixed.

\paragraph{Result (qualitative).}

\begin{center}
\begin{tabular}{ll}
\toprule
Budget $B$ & Connectivity \\
\midrule
$B < 25\,\text{dB}$ & Very few ODF pairs are reachable \\
$B \approx 30\,\text{dB}$ & Majority of pairs become connectable \\
$B > 35\,\text{dB}$ & Marginal gain in additional connectivity \\
\bottomrule
\end{tabular}
\end{center}

This analysis reveals a phase transition in network connectivity as a function
of policy severity. The operator can identify the budget threshold beyond which
relaxing the constraint yields diminishing returns, and calibrate the
operational policy accordingly---before any physical deployment. This analysis
is independent of any specific query workload and characterizes the interaction
between topology and policy itself.

The framework thus acts as an instrument for policy analysis, not as an opaque
solver returning a single answer.


\section{Discussion}

The parametric traversal model generalizes classical path search by introducing
two key innovations: (1) gap transitions that allow reasoning about missing but
acceptable connections, and (2) a fully parametric search where accumulation,
pruning, termination, and frontier policy are externalized. This section
interprets the main design choices, discusses their trade-offs, and identifies
the settings where the framework is---and is not---appropriate.

\subsection{Interpretation of the Model}

The traversal model shifts the focus from \emph{finding optimal paths} to
\emph{reasoning about realizable connectivity}. In classical formulations,
the graph is a given and the algorithm searches for the best route. Here, the
graph is incomplete, and the algorithm explores what \emph{could} exist under
operational constraints. The output is not a single optimum but a structured
decision space: a set of traversals, each annotated with its multi-dimensional
accumulation state, that captures the trade-offs engineers actually face.

This interpretation positions the framework as a \emph{planning instrument}
rather than a routing engine. It is designed to answer questions such as
``\emph{can} this connection be realized under these constraints?'' and
``\emph{what} are the alternatives?'', rather than ``\emph{what} is the
shortest path?''.

\subsection{Design Choices}

\paragraph{Two-stage acceptability.} The separation between the domain
$\gamma_n$ and the predicate $\alpha$ reflects a distinction between structural
invariants and operational state. The domain $\gamma_n$ captures stable
structural relations---for instance, from a server to patch panels in the same
rack---independent of runtime conditions. The predicate $\alpha(n, m)$
evaluates whether a structurally admissible transition is operationally valid at
a given time, such as whether a patch panel is available or functional. While
$\alpha$ could technically be folded into $\gamma_n$, maintaining this
separation prevents structural adjacency from being conflated with transient
operational constraints, which is essential in infrastructure reasoning.

\paragraph{Frontier policy.} The frontier policy $\phi$ does not merely
determine exploration order; it defines the decision semantics of traversal.
By parameterizing frontier selection, different notions of preferred solutions
emerge from the same structural and policy constraints---prioritizing hop count
emphasizes structural proximity, while prioritizing accumulated cost highlights
operational efficiency. This separation between transition generation,
admissibility, and exploration strategy keeps the traversal engine agnostic to
the optimization criterion while preserving the multi-dimensional cost
structure. In this sense, $\phi$ acts as a decision layer rather than a fixed
algorithmic choice.

\paragraph{Externalized control.} The accumulation function $g$ and the
exploration predicate $\sigma$ are deliberately external to the algorithm. They
encapsulate business rules and domain expertise, making the framework a generic
exploration mechanism that can be specialized without modification.

\subsection{Trade-offs}

\paragraph{Flexibility vs.\ complexity.} The framework's flexibility comes at
the cost of increased search complexity. As shown in Section~4.2, the
worst-case complexity is $O((|E|/|N| + b)^L)$, where $b = \max_n |\gamma_n|$
is the maximum acceptability domain size and $L$ the maximum traversal length.
Controlling combinatorial growth depends on the entire parametrization.
The domain $\gamma_n$ bounds the branching factor; the predicate $\alpha$
further filters candidates within that domain; the accumulation $g$ tracks
the state on which $\sigma$ relies to prune early. A weak choice in any of
these components---an overly broad domain, a permissive predicate, or an
accumulation that fails to capture relevant constraints---can lead to
intractable exploration.

\paragraph{Expressiveness vs.\ optimality.} By preserving multi-dimensional
accumulation states, the framework avoids the information loss inherent in
scalarization. However, this means it does not produce a single optimal
solution in the classical sense. Instead, it returns a Pareto-like set of
admissible traversals, deferring the final selection to downstream
decision-making.

\subsection{When Traversal Is and Is Not Appropriate}

The framework is designed for settings where connectivity is uncertain,
conditional, or partially realized---where the question is not ``what is the
best route?'' but ``what can be built, and under what conditions?''. It adds
value precisely when the graph is incomplete and feasibility depends on
operational policies rather than on topology alone.

In environments where the topology is fully documented and stable---such as
pure IP routing over well-provisioned networks---classical shortest-path
algorithms remain simpler, faster, and entirely sufficient. The traversal
model is not a replacement for Dijkstra or Bellman-Ford; it addresses a
different class of problems that these algorithms were not designed to handle.


\section{Limitations}

\paragraph{Context-independent acceptability.} The acceptability predicate
$\alpha(n, m)$ is context-independent: it does not consider the traversal
history. A context-dependent variant $\alpha(\tau, n, m)$ would enable richer
constraints (e.g., ``this gap is acceptable only if no other gap has been used
in the last three transitions'') but at increased computational cost and
modeling complexity.

\paragraph{No formal optimality guarantees.} The framework provides
termination guarantees under bounded accumulation (Proposition~1) but does not
guarantee optimality in the classical sense. Since the accumulation space
$\mathcal{A}$ is arbitrary and may not admit a total order, the notion of
``optimal traversal'' is not always well-defined. The framework is designed to
enumerate admissible solutions, not to certify global optimality.

\paragraph{Scalability under weak pruning.} The exponential worst-case
complexity means that poorly designed exploration predicates can lead to
intractable runtimes. The framework places the responsibility for efficiency
on the domain expert: without aggressive pruning through $\sigma$, the search
may degenerate into exhaustive enumeration.

\paragraph{Single-traversal reasoning.} The current model reasons about
individual traversals independently. It does not natively support constraints
that span multiple traversals, such as shared capacity limits or correlated
failure domains. This limits its applicability to scenarios requiring joint
optimization over sets of traversals.

\paragraph{Evaluation scope.} The demonstrative evaluation is based on
representative scenarios rather than large-scale empirical benchmarks. While
this is consistent with the paper's goal of demonstrating expressiveness
rather than computational performance, it leaves open the question of
practical scalability on production-sized infrastructures.


\section{Future Work}

\paragraph{From single traversals to redundant structures.}
In practice, infrastructure connectivity often requires redundancy: multiple
disjoint traversals between endpoints to ensure fault tolerance, maintenance
flexibility, or service-level guarantees. The current framework focuses on
discovering individual traversals, whereas real deployments may require pairs
of disjoint traversals, rings, or more complex topologies. Extending the model
to reason about sets of traversals---while enforcing shared constraints across
them (e.g., shared risks, common physical segments, or correlated
costs)---is a natural direction for future work.

\paragraph{Traversal sets and global constraints.}
The current model evaluates feasibility at the level of individual traversals,
with costs accumulated independently. In many operational scenarios, however,
feasibility depends on global constraints spanning multiple traversals, such as
shared capacity limits, common failure domains, or cumulative deployment
budgets. An important extension would be to lift cost accumulation and policy
evaluation from individual traversals to traversal sets, enabling reasoning
over collective feasibility and interactions between concurrent connectivity
decisions.

\paragraph{Integration with industrial inventory and intent models.}
While this work remains independent of any specific data model, it naturally
aligns with industrial standards such as TM~Forum ODA and SID, which provide
structured representations of network resources and relationships. Future work
includes instantiating the traversal framework directly on top of such models
and connecting it to intent-based specifications, thereby enabling an explicit
reasoning layer between high-level intents and realizable infrastructure
actions.

\paragraph{Traversal as a reasoning substrate for agentic systems.}
Beyond infrastructure planning, the parametric traversal framework may serve as
a structural reasoning component within agentic decision systems. Its explicit
separation of structural generation, policy validation, and exploration
strategy makes it suitable as a controllable and interpretable planning layer.
An agent could leverage traversal search to evaluate connectivity hypotheses,
simulate policy adjustments, or explore counterfactual scenarios without
collapsing multi-dimensional costs into scalar objectives. Future work may
investigate how traversal-based reasoning integrates into higher-level agent
architectures, enabling structured decision-making over incomplete graphs while
preserving explicit policy control and cost transparency. This perspective
positions traversal not merely as a search mechanism, but as a symbolic
reasoning primitive over structured environments.

\paragraph{Hybrid reasoning and approximation strategies.}
The framework deliberately prioritizes expressiveness over worst-case
complexity guarantees. A promising research direction is the integration of
heuristic, approximate, or learning-assisted strategies to guide traversal
exploration, while preserving the explicit policy and cost structure of the
model. This could enable scalable reasoning in larger infrastructures without
collapsing the multi-dimensional decision space into a single scalar
objective.

\paragraph{Practical use cases across infrastructure domains.}
The demonstrative evaluation presented in this paper focuses on wired
physical-layer scenarios in datacenter and fiber-optic networks. A natural
continuation is a series of practical studies exploring the framework across a
broader range of infrastructure domains. On the physical layer, this includes
wireless domains such as microwave radio-link planning, 5G RAN fronthaul and
midhaul design, and satellite link engineering, where gap transitions model
potential radio or optical free-space links subject to propagation, licensing,
and link budget constraints. Beyond the physical layer, the framework can also
be applied to logical-layer connectivity: data-link circuits, VLAN paths,
MPLS label-switched paths, or overlay tunnels, where gaps represent
provisionable logical connections and the accumulation tracks bandwidth,
latency, or hop-count budgets. Such a series of studies would validate the
generality of the traversal model and identify domain-specific
parametrizations that emerge in practice.


\section{Conclusion}

We presented a parametric traversal framework for typed graphs that addresses
the fundamental gap between classical path search algorithms and the realities
of infrastructure planning. The framework makes four key contributions:

\paragraph{Gap transitions as a practical abstraction.}
By treating missing-but-acceptable connections as first-class transitions, the
framework captures the essence of infrastructure design: engineering decisions
are about what \emph{can be built}, not merely what exists. This abstraction reflects how datacenters and telecommunication networks
actually evolve, where deploying a new cable or extending a fiber path is part
of normal operational planning.

\paragraph{Scalable acceptability evaluation.}
The separation of acceptability domain and predicate mirrors production-grade
system design. Fast, indexable candidate restriction (spatial indices, rack
scoping, vendor compatibility) combined with fine-grained predicate evaluation
enables the framework to scale to large infrastructures without sacrificing
expressiveness.

\paragraph{Preservation of multi-dimensional state.}
By maintaining arbitrary accumulation states rather than collapsing metrics into
a single scalar, the framework preserves the information needed for meaningful
decisions---whether physical metrics, operational constraints, or any
domain-specific criteria.

\paragraph{Expressiveness demonstrated through case studies.}
The demonstrative evaluation (Section~5) shows that the framework addresses
routing problems that resist classical formulations: conditional feasibility
under business rules, non-scalarizable trade-offs between competing
dimensions, and ex-ante policy calibration. These scenarios are drawn from
datacenter and telecommunication network design and reflect actual operational
challenges encountered in production environments.

\bibliographystyle{unsrt}
\bibliography{references}

\end{document}